%% file: main.tex
\documentclass[envcountsect]{llncs}

\usepackage{macros} 
\usepackage{tikz}
\usepackage{amssymb}
\usepackage{mathtools}
\usepackage{todonotes}
\usepackage{bussproofs}
\usepackage{subcaption}
\usepackage[normalem]{ulem}
\usepackage{wrapfig}
\usepackage[inline]{enumitem}
\usepackage[normalem]{ulem}
\usepackage{hyperref}

\usepackage[all]{xy}

\usetikzlibrary{matrix}

\title{Up-to Techniques for Branching Bisimilarity\thanks{The research of the second author was supported by a 
    Marie Curie Fellowship (grant code 795119). An extended abstract 
    has been published in the proceedings of SOFSEM 2020, see
    \url{https://doi.org/10.1007/978-3-030-38919-2_24}
    }}
\author{Rick Erkens\inst{1} \and
    Jurriaan Rot\inst{2,3} \and
    Bas Luttik\inst{1}}
\institute{Eindhoven University of Technology, The Netherlands \and
University College London, UK \and
Radboud University Nijmegen, The Netherlands}

\begin{document}
\maketitle
\input{abstract}
\input{introduction}
\input{preliminaries}
\input{abstractframework}
\input{bbupto}
\input{coalgebra}

\input{conc-fw}

\bibliographystyle{plain}
\bibliography{../bibliography}

\newpage

\begin{appendix}

\include{appendix}

\end{appendix}

\end{document}

%% file: abstract.tex
\begin{abstract}
Ever since the introduction of behavioral equivalences on processes one has been searching
for efficient proof techniques that accompany those equivalences.
Both strong bisimilarity and weak bisimilarity are accompanied by an arsenal of up-to techniques:
enhancements of their proof methods.
For branching bisimilarity, these results have not been established yet.
We show that a powerful proof technique is sound for branching bisimilarity
by combining the three techniques of up to union, up to expansion and up to context for
Bloom's BB cool format. We then make an initial proposal for casting the
correctness proof of the up to context technique in an abstract coalgebraic setting, 
covering branching but also $\eta$, delay and weak bisimilarity.
\end{abstract}

%% file: introduction.tex
\section{Introduction}

Bisimilarity is a fundamental notion of behavioral equivalence between processes~\cite{milner:comm}.
To prove that processes $P,Q$ are bisimilar it suffices to give a bisimulation relation $\cR$
containing the pair $(P,Q)$. But bisimulations 
can become quite large, which makes proofs long.
To remedy this issue,
up-to techniques were proposed~\cite{milner:comm,sangiorgi:proof}.
They are used, for example, in the $\pi$-calculus, where even simple properties about the replication operator
are hard to handle without them~\cite{sangiorgi:pi}, but also in automata theory~\cite{BonchiP15}
and other applications, see~\cite{pous:enhancements,BonchiPPR17} for an overview.

For weak bisimilarity the field of up-to techniques is particularly delicate.
Milner's weak bisimulations up to weak bisimilarity cannot be used to prove weak bisimilarity~\cite{sangiorgi:problem}
and the technique of up-to context is unsound for many process algebras,
most notably some that use a form of choice. Up-to techniques for weak bisimilarity
have been quite thoroughly studied (e.g.,~\cite{Pous07,pous:enhancements}). The question 
remains whether such techniques apply also to other weak
equivalences. 

In this paper, we study branching, delay and $\eta$ bisimilarity, and propose general criteria for the validity
of two main up-to techniques. We make use of the general framework
 of enhancements due to Pous and Sangiorgi~\cite{sangiorgi:proof,pous:enhancements},
 and prove that the relevant techniques are \emph{respectful}: this allows
 to modularly combine them in proofs of bisimilarity (recalled in Section~\ref{sec:abstractframework}). 
 
We start out by recasting the up-to-expansion technique, which has been proposed to remedy
certain issues in up-to techniques for weak bisimilarity~\cite{sangiorgi:problem}, 
to branching bisimilarity. 
Then, we study up-to-context techniques, which can significantly simplify bisimilarity
proofs about processes generated by transition system specifications. 
Up-to context is not sound in general, even for strong bisimilarity. For the latter,
it suffices that the specification is in the GSOS format~\cite{BonchiPPR17}. 
For weak bisimilarity one needs stronger assumptions. It was shown in~\cite{BonchiPPR17}
that Bloom's \emph{simply WB cool format}~\cite{BloomCT88:GSOS,glabbeek:cool} gives the validity of up-to context. 
We adapt this result to branching, $\eta$ and delay bisimilarity, making use of each
of the associated ``simply cool'' formats introduced by Bloom. These were introduced
to prove congruence of weak equivalences; our results extend
to respectfulness of the up-to-context technique, which is strictly stronger 
in general~\cite{pous:enhancements}.

For the results on up-to-context, we give both a concrete proof for the case of branching bisimilarity,
and a general coalgebraic treatment that covers weak, branching, $\eta$ and delay in a uniform manner (Section~\ref{sec:coalg}). 
This is based on, but also simplifies the approach in~\cite{BonchiPPR17}, by focusing on (span-based) simulations,
avoiding technical intricacies in the underlying categorical machinery. 
Our coalgebraic results are essentially about respectfulness of simulation, suitably instantiated to weak
simulations and subsequently extended to bisimulations via the general framework of~\cite{pous:enhancements}.
We conclude with some directions for future work in Section~\ref{sec:conc-fw}.

%% file: preliminaries.tex
\section{Preliminaries}\label{sec:prelims}

A Labelled Transition System (LTS) is a triple $(\states,\act,\trans)$
where $\states$ is a set of states,
$\act$ is a set of actions with $\tau \in \act$
and ${\trans} \subseteq \states \times \act \times \states$ is a set of transitions.
We denote a transition $(P,\alpha,P')$ by $P \step{\alpha} P'$.
For any $\alpha$ we consider $\step{\alpha}$ a binary relation on $\states$.
With this in mind let $\tausteps$ denote the transitive reflexive closure of $\taustep$.
By $P \step{(\alpha)} P'$ we mean that $P \step{\alpha} P'$ or $\alpha=\tau$ and $P=P'$.
The capital letters $P,Q,X,Y,Z$ range over elements of $\states$.
The letters $\alpha, \beta$ denote arbitrary elements from $\act$
and with lowercase letters $a$ we denote arbitrary elements of $\act  \backslash \{\tau\}$, so the action $a$ is not a silent action.

The set of relations between sets $X$ and $Y$ is denoted by $\Rel_{X,Y}$;
when $X=Y$ we denote it by $\Rel_X$, ranged over by $\cR, \cS$.
Relation composition is denoted by $\cR \relcmp \cS = \{(P,Q) \mid \exists X. \, P \rcR X
    \text{ and }X \rcS Q\}$, or simply by $\cR \cS$. 
For any set $X$, the partial order $(\Rel_X, \subseteq)$ forms a complete lattice,
where the join and meet are given by union $\bigcup X$ and intersection $\bigcap X$
respectively. 
%
A function $f \colon \Rel_X \rightarrow \Rel_X$ is monotone iff $\cR \subseteq \cS$ implies $f(\cR) \subseteq f(\cS)$.
The set $[\Rel_X \rightarrow \Rel_X]$ of such monotone functions is again a complete lattice, ordered by
pointwise inclusion, which we denote by $\leq$. Thus, join and meet are 
pointwise: $\bigvee F = \lambda \cR.\bigcup\{f(\cR) \mid f \in F\}$
and $\bigwedge F = \lambda \cR.\bigcap\{f(\cR) \mid f \in F\}$.

\paragraph{Bisimulation.}
Consider the function
$\brsmono(\rcR)$ = $\{(P,Q) \mid$ for all $P'$ and for all $\alpha$, if $P \step{\alpha} P'$ then there exist $Q', Q''$ s.t. 
								 $Q \tausteps Q' \step{(\alpha)} Q''$ and $P \rcR Q'$ and $P'\rcR Q'' \} $. 
We say that $\cR$ is a \emph{branching simulation} if $\cR \subseteq \brsmono(\cR)$.
Moreover we define $\brmono = \brsmono \meet (\rev \circ \brsmono \circ \rev)$ where
$\rev(\cR) = \{(Q,P) \mid P \rcR Q\}$ and say that
$\cR$ is a \emph{branching bisimulation} if $\cR \subseteq \brmono(\cR)$. 
Since $\brsmono$ and hence $\brmono$ are monotone and $(\relp, \subseteq)$ is a complete lattice,
$\brmono$ has a greatest fixed point.
We denote it by $\brbis$ and refer to it as \emph{branching bisimilarity}.
To prove $P \rbrbis Q$, it suffices to provide a relation $\cR$
that contains the pair $(P,Q)$ and show that $\cR \subseteq \brmono(\cR)$; the latter implies $\cR \subseteq {\brbis}$. 
Up-to techniques strengthen this principle (Section~\ref{sec:abstractframework}).
%

Delay (bi)similarity is defined analogously through the function $\dsmono$, defined as
$\brsmono$ but dropping the condition $P \rcR Q'$. 
Weak simulations are defined using the map
$\wsmono(\cR) = \{(P,Q) \mid$ for all $P'$ and for all $\alpha$, if $P \step{\alpha} P'$ then there exist $Q', Q'',Q'''$
                                such that $Q \tausteps Q' \step{(\alpha)} Q''\tausteps Q'''$ and $P'\rcR Q''' \}$.
Finally, for $\eta$ simulation, we have $\etasmono$, defined as $\wsmono$ but adding the requirement
$P \rcR Q'$.

Intuitively, the four notions of bisimilarity defined above vary
  in two dimensions: first, \emph{branching} and \emph{delay}
  bisimilarity consider internal activity (represented by
  $\tau$-steps) only before the observable step, whereas $\eta$ and
  \emph{weak} bisimilarity also consider internal activity after the
  observable step; second, \emph{branching} and $\eta$ bisimilarity
  require that the internal activity does not incur a change of state,
  whereas for \emph{delay} and \emph{weak} this is not required.

%
%



\paragraph{GSOS and Cool Formats.}
%
\emph{GSOS} is a rule format that guarantees strong bisimilarity
to be a congruence~\cite{BloomCT88:GSOS}. Bloom introduced \emph{cool languages} 
as restrictions of GSOS, forming suitable formats for 
weak, branching, $\eta$ and delay bisimilarity~\cite{bloom:structural}.


A signature $\sig$ is a set of operators that each have an arity
denoted by $\arity(\sigma)$.
We assume a set of variables $\vars$ and denote the set of terms over a signature $\sig$ by $\terms(\sig)$.
For a term $t$ we denote the set of its variables by $\varsof(t)$.
A term $t$ is closed if $\varsof(t)=\emptyset$.
A \emph{substitution} is a partial function ${\rho:\vars \rightharpoonup\terms(\sig)}$.
We denote the application of a substitution to a term $t$ by $t^\rho$.
A substitution is \emph{closed} if $\rho$ is a total function such that $\rho(x)$ is closed
for all $x\in\vars$.

\begin{definition}\label{def:gsos}
A positive GSOS language is a tuple $(\sig,\rules)$ where $\sig$ is a signature and
$\rules$ is a set of transition rules of the form $\frac{H}{\sigma(x_1,\dots,x_{\arity(\sigma)}) \step{\alpha} t}$
where $t$ is a term, $x_1,\dots,x_{\arity(\sigma)}$ are distinct variables and $H$ is a set of premises such that
each premise in $H$ is of the form $x_i \step{\beta} y_i$
    where the left-hand side $x_i$ occurs in $x_1,\dots,x_{\arity(\sigma)}$;
the right-hand sides $y_i$ of all premises are distinct;
the right-hand sides $y_i$ of all premises do not occur in $x_1,\dots,x_{\arity(\sigma)}$; 
the target $t$ only contains variables that occur in the premises or in the source.
\end{definition}

The (not necessarily positive) GSOS format also allows negative premises,
that is, premises of the form $\smash{x_i \not\step{\beta}}$.
In this paper we do not consider those.

  An \emph{LTS algebra} for a signature $\sig$
  consists of an LTS  $(\states,\act,\trans)$ together with a
  $\sig$-indexed family of mappings on $\states$ of corresponding arity. We
  denote the mapping associated with an element of
  $\sig$ by the same symbol, i.e., for all $\sigma\in\Sigma$ there
  is a map $\sigma: \states^{\arity(\sigma)}\rightarrow\states$.
  If $t\in\terms(\sig)$ and
  $\rho:\vars\rightarrow\states$ is an assignment of states to
  variables, then we denote by $t^\rho$ the interpretation of $t$ in
  $\states$.

  Now, let $\lang=(\sig,\rules)$ be a GSOS language. Then an LTS
  algebra for $\sig$ is a \emph{model} for $\lang$ if it satisfies the
  rules in $\rules$, i.e., if for every rule
    $\frac{H}{\sigma(x_1,\dots,x_n)\step{\alpha}t} \in \rules$
    and for every assignment $\rho$ we have that
    whenever
      $\rho(x_i)\step{\beta_i}\rho(y_i)$ for every premise $x_i
      \step{\beta_i} y_i \in H$
    then also
      $\sigma(\rho(x_1),\dots,\rho(x_n))\step{\alpha} t^{\rho}$.
 
  The \emph{canonical model} for a GSOS language $\lang=(\sig,\rules)$ has
   as states the set of closed $\sig$-terms and for all closed
   terms $P$ and $P'$, a transition
   $P\step{\alpha}P'$ if, and only if, there is a rule
     $\frac{H}{\sigma(x_1,\dots,x_n)\step{\alpha}t} \in \rules$
     and a substitution $\rho$ such that
       $P = \sigma(\rho(x_1),\dots,\rho(x_n))$,
       $P' = t^\rho$,
      and $\rho(x_i)\step{\beta_i}\rho(y_i)$ for all premises $x_i
      \step{\beta_i} y_i \in H$. The mapping associated with an
      $n$-ary element $\sigma\in\sig$ maps every sequence
      $t_1,\dots,t_n$ for closed terms to the closed term $\sigma(t_1,\dots,t_n)$. 

Bloom's \emph{cool formats} \cite{bloom:structural} rely on some auxiliary notions. 
  A rule of the form $\frac{x_i \step{\tau} y_i}{\sigma(x_1,\dots,x_n) \step{\tau}
  \sigma(x_1,\dots,y_i,\dots,x_n)}$ is called a \emph{patience rule} for
the $i$th argument of $\sigma$.
 A rule is \emph{straight} if the left-hand sides of all premises
  are distinct. A rule is \emph{smooth} if, moreover, 
    no variable occurs both in the target and the left-hand
    side of a premise.
 The $i$th argument of $\sigma \in \Sigma$ is \emph{active} if there is a rule $\frac{H}{\sigma(x_1,\dots,x_n)\step{\alpha}t}$
  in which $x_i$ occurs at the left-hand side of a premise.
A variable $y$ is \emph{receiving} in the target $t$ of a rule $r$ in
    $\lang$  if it is the right-hand side of a
    premise of $r$. The $i$th argument of $\sigma \in \Sigma$ is \emph{receiving} if there is a variable $y$ and a target $t$ of a rule in $\lang$
    s.t.\ $y$ is receiving in $t$, $t$ has a subterm
    $\sigma(v_1,\dots,v_n)$ and $y$ occurs in $v_i$.

For instance, CCS~\cite{milner:comm} has the rule $\frac{x_1 \step{\alpha} y_1}{x_1 + x_2 \step{\alpha} y_1}$
for the binary choice operator $+$.
The first argument is active,
but the semantics does not allow a patience rule for it.
The issue can be mitigated by guarded sums, replacing choice 
by infinitely many rules of the form $\Sigma_{i\in I}: \alpha_i.x_i \step{\alpha_i} x_i$.
These rules have no premises; therefore there are no active arguments and no patience rule is needed.

\begin{definition}\label{def:formats}
A language $\lang=(\sig,\rules)$ is \emph{simply WB cool} if it is positive GSOS and
\begin{enumerate*}
\item\label{item:coolstraight} all rules in $\lang$ are straight;
\item\label{item:cooltaupatience} only patience rules have $\tau$-premises;
\item\label{item:coolactive} for each operator every active argument has a patience rule;
\item\label{item:coolreceivingpatience} every receiving argument of an operator has a patience
    rule; and
\item\label{item:coolsmooth}  all rules in $\lang$ are smooth.
\end{enumerate*}
The language $\lang$ is \emph{simply BB cool} if it satisfies
   \ref{item:coolstraight}, \ref{item:cooltaupatience}, and
  \ref{item:coolactive}.
  It is \emph{simply HB cool} if it satisfies \ref{item:coolstraight}, \ref{item:cooltaupatience},
  \ref{item:coolactive}, and \ref{item:coolreceivingpatience}. It is \emph{simply DB cool} if it
satisfies \ref{item:coolstraight}, \ref{item:cooltaupatience},
  \ref{item:coolactive}, and \ref{item:coolsmooth}.
\end{definition}



In \cite{glabbeek:cool}, van Glabbeek presents four lemmas, labelled BB, HB, DB and WB, respectively, that are instrumental for proving that branching, $\eta$, delay and
  weak bisimilarity are congruences for the associated variants of
  cool languages. In \cite{glabbeek:cool} these lemmas are established
  for the canonical model, but they have straightforward
  generalisations to arbitrary models; these generalisations will be instrumental  for our results
  in Sections~\ref{sec:bb-up-to} and \ref{sec:coalg}. We only present the
  generalisations of WB and BB here; the generalisations of HB and DB
  proceed analogously.

\begin{lemma}\label{lem:superrobweak}
    Let $\lang$ be a simply WB cool language,
    let $(\states,\act,\trans)$ be a model for $\lang$,
    let $\eta: \vars\rightarrow\states$ be an assignment
    and let $\frac{H}{\sigma(x_1,\dots,x_n)\step{\alpha} t}$ be a rule in
    $\lang$.
    If for each $\smash{x\step{\beta}y}$ in $H$ we have
    $\smash{\eta(x)\tausteps\step{(\beta)}\tausteps\eta(y)}$, then
    $\smash{\sigma(x_1,\dots,x_n)^\eta \tausteps\step{(\alpha)}\tausteps t^\eta}$.
\end{lemma}

\begin{lemma}\label{lem:superrob}
Let $\lang$ be a simply BB cool language, let
    $(\states,\act,\trans)$ be a model for $\lang$,
and let $\frac{\{x_i \step{\beta_i} y_i \mid i\in I\}}{\sigma(x_1,\dots,x_n)\step{\alpha} t}$
be some rule in $\lang$.
If $\eta, \theta: \vars\rightarrow\states$ are assignments s.t.\
for all $i\in I$ it holds that $\smash{\eta(x_i)\tausteps\theta(x_i)\step{(\beta_i)}\theta(y_i)}$ and
for every $x\not\in\{x_i, y_i\mid i\in I\}$ we have $\eta(x)=\theta(x)$,
then $\sigma(x_1,\dots,x_n)^\eta \tausteps \sigma(x_1,\dots,x_n)^\theta \step{(\alpha)} t^\theta$.
\end{lemma}

%% file: abstractframework.tex
\section{The abstract framework for bisimulations}\label{sec:abstractframework}

We recall the lattice-theoretical framework of up-to techniques
proposed by Pous and Sangiorgi~\cite{pous:coin}, which allows to obtain
enhancements of branching bisimilarity and other coinductively defined relations
in a modular fashion. 
Throughout this section, let $f,b,s \colon \Rel_X \rightarrow \Rel_X$ be monotone maps. 

We think of $\gfp(b)$ as the coinductive object of interest (e.g., bisimilarity);
then, to prove $(P,Q) \in \gfp(b)$ it suffices to prove $(P,Q) \in \cR$ for some 
$\cR \subseteq b(\cR)$ (e.g., a bisimulation). 
The aim of using up-to techniques is to alleviate this proof obligation,
by considering an additional map $f$, and proving instead
that $\cR \subseteq b(f(\cR))$; such a relation is called
a \emph{$b$-simulation up to $f$} (e.g., a bisimulation up to $f$). 
Typically, this map $f$ will increase the argument relation. 
Not every function $f$ is suitable as an up-to technique: it should be sound.

\begin{definition}
We say that $f$ is $b$-sound if $\gfp(b \circ f) \subseteq \gfp(b)$.
\end{definition}
When one proves $\cR\subseteq b(f(\cR))$
it follows that $\cR \subseteq \gfp(b \circ f)$.
Soundness is indeed the missing link to
conclude $\cR \subseteq \gfp(b)$.
Unfortunately the composition of two $b$-sound functions is not $b$-sound in general~\cite[Exercise~6.3.7]{pous:enhancements}.
To obtain compositionality we use the stronger notion of respectfulness.

\begin{definition}
A function $f$ is $b$-respectful if $f\circ (b\meet\id) \below (b\meet\id) \circ f$.
\end{definition}

This originates from Sangiorgi~\cite{sangiorgi:proof},
and was used to prove that up-to context is sound for
strong bisimilarity, for faithful contexts.
Lemma~\ref{lem:compositionality} states that respectful functions are sound, and gives
methods to combine them.
It summarises certain results from~\cite{pous:coin,pous:enhancements}
about \emph{compatible} functions: $f$ is $b$-compatible
if $f \circ b \below b \circ f$. Thus, respectfulness simply means
$b \meet \id$-compatibility. While compatibility is stronger than respectfulness,
this difference disappears if we move to the \emph{greatest} compatible function,
given as the join of all $b$-compatible functions. 

\begin{lemma}\label{lem:compositionality}
Consider the \emph{companion} of $b$,
defined by $\bro =\bigvee \{f \mid f \circ b \below b \circ f\}$.
\begin{enumerate}
\item for all respectful functions $f$ it holds that $f\below \bro$;
\item for all sets $F$ such that $f\below\bro$ for every $f\in F$, we have $\bigvee F \below \bro$;
\item for any two functions $f,g\below\bro$ it holds that $g \circ f \below \bro$;
\item if $\cS \subseteq b(\cS)$ then, for $\utS$ the constant-to-$\cS$ function, $\utS \leq \bro$;
\item if $f \below \bro$ then $f(\gfp(b)) \subseteq \gfp(b)$;
\item if $f \below \bro$ then $f$ is $b$-sound.
\end{enumerate}
\end{lemma}
Lemma~\ref{lem:compositionality} is used to obtain powerful proof techniques for branching bisimilarity
and other coinductive relations.
If $f$ is below the companion $\bro$, it can safely be used as an up-to technique; moreover,
such functions combine well, via composition and union. 
The above lemma gives some basic up-to techniques for free: for instance,
the function $f(\cR) = \cR \cup \gfp(b)$ is below $\bro$ (for any $b$). 
We will focus on up-to-expansion and up-to-context. Especially
the latter requires more effort to establish, but can drastically alleviate the effort in proving 
bisimilarity. 


%


We conclude with two useful lemmas. The first states that
for symmetric techniques it suffices to prove respectfulness for similarity,
and the second is a proof technique for respectfulness (and, in fact, the original characterisation).

\begin{lemma}\label{lem:simulationisenough}
Let $b = s \meet (\rev \circ s\circ \rev)$.
If $f$ is symmetric (i.e. $f = \rev \circ f \circ \rev$) and $s$-respectful, then $f$ is $b$-respectful.
\end{lemma}
%

\begin{lemma}\label{lem:reformulation}
The function $f$ is $b$-respectful if and only if for all $\cR, \cS$ we have that $\cR \subseteq \cS$ and $\cR \subseteq b(\cS)$
implies $f(\cR) \subseteq b(f(\cS))$.
\end{lemma}

%% file: bbupto.tex
\section{Branching bisimilarity: expansion and context}
\label{sec:bb-up-to}

\paragraph{Up-to expansion.}
The first up-to technique for strong bisimilarity was reported by Milner \cite{milner:comm}.
It is based on the enhancement function $\utsRs$
where $\sbis$ denotes strong bisimilarity.
It is well known that a similar enhancement function $\utwRw$ is unsound for weak bisimilarity \cite{sangiorgi:problem,pous:enhancements},
and the same counterexample shows that the enhancement function $\utbRb$ is unsound
for branching bisimilarity: the relation $\{(\tau.a, 0)\}$ on CCS processes~\cite{milner:comm} 
is a branching bisimulation up to $\utbRb$, using that
$a \brbis \tau.a$, but clearly
$\tau.a$ is not branching bisimilar to $0$.
The function $\utsRs$ is $\brmono$-respectful. But it turns out that one can do slightly better,
using an efficiency preorder called expansion \cite{kumar:efficiency,pous:enhancements}.
We proceed to define such a preorder for branching bisimilarity and show that it results
in a more powerful up-to technique than strong bisimilarity.


\begin{definition}
Consider the function $\emonogt \colon \relp \rightarrow \relp$ defined as
$\emonogt(\cR) = \{ (P, Q) \mid$ for all $P'$ and all $\alpha$, if $P \step{\alpha} P'$ then
     there exists $Q'$ such that $Q \step{(\alpha)} Q'$ and $P' \rcR Q'$; and
     for all $Q'$ and all $\alpha$, if $Q \step{\alpha} Q'$ then
     there exist $P',P''$ such that $P \tausteps P' \step{\alpha} P''$
     with $P' \rcR Q$ and $P'' \rcR Q' \}.$
We say $\cR$ is a \emph{branching expansion} if $\cR \subseteq \emonogt(\cR)$.
Denote $\gfp(\emonogt)$ by $\gexp$. 
\end{definition}

Informally $P \rgexp Q$ means that $P$ and $Q$ are branching bisimilar and
$P$ always performs at least as many $\tau$-steps as $Q$.
Similar notions of expansion can be defined for $\eta$ and delay bisimilarity.
Examples are at the end of this section.


\begin{lemma}\label{lem:expansionrespectful}
The function $\uteRe$ is $\brmono$-respectful.
\end{lemma}

The proof of Lemma~\ref{lem:expansionrespectful} is routine if we use Lemmas~\ref{lem:simulationisenough}
and \ref{lem:reformulation}.
It suffices to show that if $\cR \subseteq \brsmono(\cS)$ and $\cR\subseteq\cS$ then
${\gexp\rcR\lexp} \subseteq {\brsmono(\gexp\rcS\lexp)}$.
This inclusion can be proved by playing the branching simulation game on the pairs in $\gexp\rcR\lexp$.


\paragraph{Up-to context.}
Next, we consider LTSs generated by GSOS languages.
Here, an up-to-context technique enables us to use congruence properties
of process algebras in the bisimulation game: it suffices
to relate terms by finding a mutual context for both terms.
We show that if $\lang$ is a language
in the simply BB cool format,
then the closure w.r.t.\ $\lang$-contexts is
$\brmono$-respectful. 

\begin{definition}\label{def:ctx}
Let $\lang=(\sig,\rules)$ be a positive GSOS language and let $\cR$ be a relation on closed $\lang$-terms.
The closure of $\cR$ under $\lang$-contexts is denoted by $\ctxL(\cR)$ and is defined as
the smallest relation that is closed under the following inference rules:
$
\frac{P \rcR Q}{P \rctxLR Q}
$
and 
$
\frac{P_1 \rctxLR Q_1 \quad \dots \quad P_{\arity(\sigma)} \rctxLR Q_{\arity(\sigma)}}
{\sigma(P_1,\dots,P_{\arity(\sigma)}) \rctxLR \sigma(Q_1,\dots,Q_{\arity(\sigma)})}
$.
\end{definition}

\begin{theorem}\label{thm:br-up-to-ctx}
Let $\lang$ be a simply BB cool language.
Then $\ctxL$ is $\brmono$-respectful.
\end{theorem}

For the proof, we use Lemma~\ref{lem:reformulation} and show that if $\rcR \subseteq \brmono(\rcS)$ 
and $\rcR \subseteq \rcS$ then $\ctxL(\cR) \subseteq \brmono(\ctxL(\cR))$. The proof is by induction on
elements of $\ctxL(\cR)$, using Lemma~\ref{lem:superrob}, which essentially states that a suitable saturation
of the canonical model of $\lang$ (Section~\ref{sec:prelims}) is still a model of $\lang$. 
This is generalised in Section~\ref{sec:coalg}.

The following two examples use a variant of CCS \cite{milner:comm}
  with replication (!); we refer to \cite{pous:enhancements} for its syntax and operational semantics.
\begin{example}\label{ex:uptoeCRe}
We show that ${!\tau.(a\pc\compl{a})} \rbrbis {!(\tau.a+\tau.\compl{a})}$.
Consider the relation $\cR$ containing just the single pair of processes.
It suffices to prove that $\cR$ is a branching bisimulation up to $\uteCRe$
since both $\ctxL$ and $\uteRe$ are $\brmono$-respectful.
In the proof one can use properties for strong bisimilarity like $!P\pc P \rsbis P$
and $P\pc Q \rsbis Q \pc P$.
Since ${\sbis} \subseteq {\gexp}$ these laws also apply to expansion.
Then the expansion law $P\pc \tau.Q \rgexp P\pc Q$ ensures that $\cR$ suffices.
\end{example}

\begin{example}\label{ex:uptosCRs}
We show that $!(a + b) \rbrbis {!\tau.a}\pc{!\tau.b}$.
The relation $\cR$ containing just the single
pair of processes is a branching bisimulation up to $\utsCRs$.
This is sufficient: 
since $\uteRe$ is $\brmono$-respectful and $\utsRs \below \uteRe$,
the function $\utsRs$ is below the companion of $\brmono$, 
and therefore it can be combined with $\ctxL$ to obtain a $\brmono$-sound technique.
\end{example}

A similar result as Theorem~\ref{thm:br-up-to-ctx} is established for weak bisimilarity in~\cite{BonchiPPR17}.
In fact, one can use the lemmas at the end of Section~\ref{sec:prelims} to treat $\eta$ and delay bisimilarity as well. 
We develop a uniform approach in the following section. 

%% file: coalgebra.tex
\section{Respectfulness of up-to context: coalgebraic approach}\label{sec:coalg}

We develop conditions for respectfulness of contextual closure
that instantiate to variants for branching, weak,
$\eta$ and delay bisimilarity. In each case, the relevant condition 
is implied by the associated simply cool GSOS format. 

The main step is that contextual closure is respectful
for \emph{similarity}, for a relaxed notion of models of positive GSOS specifications. 
The case of weak, branching, $\eta$ and delay 
are then obtained by considering simulations between LTSs and appropriate saturations thereof.\footnote{Note that this is fundamentally different from reducing weak bisimilarity to strong
bisimilarity on a saturated transition system; there, a challenging transition is weak as well. Here,
instead, strong transitions are answered by weak transitions.}
We use the theory of coalgebras; in particular, the respectfulness
result for simulations is phrased at an abstract level. 
We assume familiarity with basic notions in category theory. 
Further, due to space constraints, we only
report basic definitions; see, e.g.,~\cite{J2016,Rutten00} for details.

The abstract results in this section are inspired by, and close to,
the development in~\cite{BonchiPPR17}. Technically, however, we simplify in two ways:
(1) focusing on simulations rather than on (weak) bisimulations directly 
through functor lifting in a fibration;
and (2) avoiding the technical
sophistication that arises from the combination of fibrations and orderings, by using 
a (simpler) span-based approach in the proofs. Still, we use a number of results from~\cite{BonchiPPR17},
connecting monotone GSOS specifications to distributive laws. 
The cases of branching, $\eta$ and delay bisimilarity, which we treat here, were left as future work in~\cite{BonchiPPR17}.
Note that we do not propose a general coalgebraic theory of weak bisimulations, as introduced, e.g., in~\cite{Brengos13},
but focus on LTSs, which are the models of interest here.


\emph{Coalgebra.}
We denote by $\Set$ the category of sets and functions. 
Given a functor $B \colon \Set \rightarrow \Set$, a \emph{$B$-coalgebra} is a pair $(X,f)$
where $X$ is a set and $f \colon X \rightarrow B(X)$ a function. A \emph{coalgebra homomorphism}
from a $B$-coalgebra $(X,f)$ to a $B$-coalgebra $(Y,g)$ is a map $h \colon X \rightarrow Y$ 
such that $g \circ h = Bh \circ f$. 

Let $\act$ be a fixed of labels with $\tau \in \act$.
Labelled transition systems
are (equivalent to) coalgebras for the functor $B$ given by $B(X) = (\pow X)^{\act}$.
Indeed, a $B$-coalgebra
consists of a set of states $X$ and a map $f \colon X \rightarrow (\pow X)^{\act}$ mapping
a state $x \in X$ to its outgoing transitions; we write
$x \xrightarrow{\alpha}_f y$ or simply $x \xrightarrow{\alpha} y$ for $y \in f(x)(\alpha)$.  
In this section we mean coalgebras for this functor, when referring to LTSs. 
The notations $\Longrightarrow_f$ and
$\smash{\xrightarrow{(\alpha)}_f}$, defined in Section~\ref{sec:prelims},
are used as well.

To define (strong) bisimilarity of coalgebras we make use of \emph{relation lifting}~\cite{J2016},
which maps a relation $R \subseteq X \times Y$ to a relation $\Rel(B)(R) \subseteq BX \times BY$. 
This is given by
$
\Rel(B)(R) = \{(u,v) \mid \exists z \in B(R). B(\pi_1)(z) = u \text{ and } B(\pi_2)(z) = v \}
$.
Now, given $B$-coalgebras $(X,f)$ and $(Y,g)$, a relation $R \subseteq X \times Y$ is a bisimulation
if for all $(x,y) \in R$, we have $f(x) \mathrel{\Rel(B)(R)} g(y)$. 
In case of labelled transition systems, this 
amounts to the standard notion of strong bisimilarity. 

\emph{Algebra.}
An algebra for a functor $H \colon \Set \rightarrow \Set$ is
a pair $(X,a)$ where $X$ is a set and $a \colon H(X) \rightarrow X$ a function. An algebra morphism
from $(X,a)$ to $(Y,b)$ is a map $h \colon X \rightarrow Y$ such that $h \circ a = b \circ Hh$. 
While coalgebras are used here to represent variants of labelled transition systems, 
we will also make use of algebras, to speak about operations in process calculi. 
In order to do so, we first show how to represent a signature $\Sigma$ as a functor $H_\Sigma \colon \Set \rightarrow \Set$, such that 
$H_\Sigma$ algebras are interpretations of the signature $\Sigma$. 
Given $\Sigma$, this functor $H_\Sigma$ is 
defined by:
$
H_\Sigma(X) = \coprod_{\sigma \in \Sigma} \{\sigma\} \times X^{\arity(\sigma)}
$.
On maps $f \colon X \rightarrow Y$, $H_\Sigma$ is defined pointwise, i.e., $H_\Sigma(f)(\sigma(x_1, \ldots, x_{\arity(\sigma)})) = 
\sigma(f(x_1), \ldots, f(x_{\arity(\sigma)}))$.

We denote by $T_\Sigma \colon \Set \rightarrow \Set$ the \emph{free monad} of $H_\Sigma$. 
Explicitly, $T_\Sigma(X)$ is the set of terms over $\Sigma$ with variables in $X$, as generated by the grammar
$t~{::=}~x \mid \sigma(t_1, \ldots, t_{\arity(\sigma)})$
where $x$ ranges over $X$ and $\sigma$ ranges over $\Sigma$.
In particular, $T_\Sigma(\emptyset)$ is the set of closed terms. 
The set $T_\Sigma(X)$ is the carrier of a \emph{free algebra} $\kappa_X \colon H_\Sigma T_\Sigma(X) \rightarrow T_\Sigma(X)$:
there is an arrow $\eta_X \colon X \rightarrow T_\Sigma(X)$ (the unit of the monad $T_\Sigma$) 
such that, for every algebra $b \colon H_\Sigma(Y) \rightarrow Y$ and map $f \colon X \rightarrow Y$,
there is a unique algebra homomorphism $f^\sharp \colon T_\Sigma(X) \rightarrow Y$ s.t.\ $f^\sharp \circ \eta_X = f$. 
In particular, we write $b^{*} \colon T_\Sigma(Y) \rightarrow Y$ for $\id_Y^\sharp$. 
Intuitively, $b^{*}$ inductively extends the algebra structure $b$ on $Y$ to terms over $Y$. 

\emph{Simulation of coalgebras.}
We recall how to represent simulations~\cite{JacobsH03}, based on ordered functors. 
This enables speaking about weak simulations (Section~\ref{sec:coalg-weak}).
As before, by Lemma~\ref{lem:simulationisenough}, relevant respectfulness results extend to bisimulations. 

An \emph{ordered functor} is a functor $B \colon \Set \rightarrow \Set$ together with, for every set $X$,
a preorder ${\sqsubseteq_{BX}} \subseteq BX \times BX$ such that, for every map $f \colon X \rightarrow Y$,
$Bf \colon BX \rightarrow BY$ is monotone. Equivalently, it is a functor
$B$ that factors through the forgetful functor $U \colon \PreOrd \rightarrow \Set$ from the category 
of preorders and monotone maps. 
For maps $f,g \colon X \rightarrow BY$, we write $f \sqsubseteq_{BY} g$ for pointwise inequality, i.e.,
$f(x) \sqsubseteq_{BY} g(x)$ for all $x \in X$. 
Throughout this section we assume $B$ is ordered. 


To define simulations, we recall from~\cite{JacobsH03} the \emph{lax relation lifting} $\Rel_\sqsubseteq(B)$, 
defined on a relation $R \subseteq X \times Y$ as
$
	\Rel_\sqsubseteq(B)(R) = {{\sqsubseteq_{BX}}  \relcmp \Rel(B)(R) \relcmp {\sqsubseteq_{BY}}}
        $.
\begin{definition}\label{def:coalg-sim}
Let $(X,f)$ and $(Y,g)$ be $B$-coalgebras. Define the following monotone operator $s \colon \Rel_{X,Y} \rightarrow \Rel_{X,Y}$
by
$
	s(R) = (f  \times g)^{-1} (\Rel_\sqsubseteq(B)(R)) 
$.
A relation $R \subseteq X \times Y$ is called a \emph{simulation} if 
it is a post-fixed point of $s$.  
\end{definition}
\begin{example}\label{ex:ordered-lts}
	The functor $B(X) = (\pow X)^{\act}$ is ordered, with $u \sqsubseteq_{BX} v$ iff $u(a) \subseteq v(a)$ for all $a \in \act$. 
	The associated lax relation lifting maps $R \subseteq X \times Y$ 
	to $\Rel_\sqsubseteq(B)(R) = \{(u,v) \mid \forall a \in \act. \, \forall x \in u(a). \, \exists y \in v(a). \, (x,y) \in R\}$. 
	A relation $R \subseteq X \times Y$ between (the underlying state spaces of) LTSs is a simulation in the sense
	of Definition~\ref{def:coalg-sim} iff it is a simulation 
	in the standard sense: for all $(x,y) \in R$: if $x \xrightarrow{\alpha} x'$ then $\exists y'. \, y \xrightarrow{\alpha} y'$ and 
	$(x',y') \in R$. 
\end{example}

\subsection{Abstract GSOS specifications and their models}

  \begin{wrapfigure}[5]{r}{0pt}
\begin{minipage}{20em} \vspace{-3.5em}
$$
\xymatrix@R=0.5cm{
H_\Sigma(X) \ar[d]_{a} \ar[r]^-{H_\Sigma \langle f, \id \rangle}
	& H_\Sigma (BX \times X) \ar[r]^-{\lambda_X}
	& B T_\Sigma(X) \ar[d]^{B a^{*}} \\
X \ar[rr]^-{f} 
	& & BX  
}
$$
\end{minipage}
\end{wrapfigure}
An \emph{abstract GSOS specification}~\cite{TP97} is a natural transformation of the form 
$\lambda \colon H_\Sigma(B \times \Id) \Rightarrow BT_\Sigma$. 
Let $X$ be a set, let $a \colon H_\Sigma(X) \rightarrow X$ be an
  algebra, and let $f \colon X \rightarrow BX$ be a coalgebra; the
  triple $(X,a,f)$ is a \emph{$\lambda$-model} if the
  diagram on the right commutes.

In our approach to proving the validity of up-to techniques for weak similarity,
it is crucial to relax the notion of $\lambda$-model to a \emph{lax model}, following~\cite{BonchiPPR17}.
A triple $(X,a,f)$ as above is a \emph{lax $\lambda$-model} if we 
have that $f \circ a \sqsubseteq_{BX} Ba^{*} \circ \lambda_X \circ H_\Sigma \langle f, \id \rangle$, 
and an \emph{oplax $\lambda$-model} if, conversely, 
$f \circ a \sqsupseteq_{BX} Ba^{*} \circ \lambda_X \circ H_\Sigma \langle f, \id \rangle$. 
Since $\sqsubseteq_{BX}$ is a preorder, $(X,a,f)$ is a $\lambda$ model iff it is both a lax and an oplax model. 

Taking the algebra $\kappa_\emptyset \colon H_\Sigma T_\Sigma \emptyset \rightarrow T_\Sigma \emptyset$
on closed terms, there is a unique coalgebra structure $f \colon T_\Sigma \emptyset \rightarrow BT_\Sigma \emptyset$
turning $(T_\Sigma \emptyset, \kappa_\emptyset, f)$ into a $\lambda$-model. We sometimes refer to this coalgebra
structure as the \emph{operational model} of $\lambda$. 

We say $\lambda$ is \emph{monotone} if for
each component $\lambda_X$, we have
$$
\frac{
	u_1 \sqsubseteq_{BX} v_1 \qquad \ldots \qquad u_n \sqsubseteq_{BX} v_n
}{
	\lambda_X(\sigma((u_1, x_1), \ldots, (u_n, x_n))) \sqsubseteq_{BT_\Sigma X} 
	\lambda_X(\sigma((v_1, x_1), \ldots, (v_n, x_n)))
}
$$
for every operator $\sigma \in \Sigma$, elements $u_1, \ldots u_n, v_1, \ldots, v_n \in BX$ and
$x_1, \ldots, x_n \in X$, with $n = \arity(\sigma)$. Informally, if premises have `more behaviour' (e.g.,
more transitions) then we can derive more behaviour from the GSOS specification.

\begin{example}
If $BX = (\pow X)^{\act}$, then a monotone $\lambda$ corresponds to a positive GSOS specification (Definition~\ref{def:gsos}). 
In that case, an algebra $a \colon H_\Sigma(X) \rightarrow X$ together with
a $B$-coalgebra (i.e., LTS) is a $\lambda$-model if, 
for every $P \in X$, we have that 
$P \xrightarrow{\alpha} P'$ iff there is a rule $\frac{H}{\sigma(x_1,\dots,x_n)\step{\alpha}t}$
    and a map $\rho \colon V \rightarrow X$ (with $V$ the set of
    variables occurring in the rule) such that
    $P = a(\sigma(\rho(x_1),\dots,\rho(x_n)))$,
    $P' = \rho^\sharp(t)$ (recall that $\rho^{\sharp}$
      denotes the unique algebra homomorphism associated with $\rho$)
    and for all premises $x_i \step{\beta_i} y_i \in H$ we have that
    $\rho(x_i)\step{\beta_i}\rho(y_i)$.
This coincides with the interpretation in Section~\ref{sec:prelims}.
A lax model only asserts the implication from right to left (transitions are closed under application of 
rules) and an oplax model asserts the converse (every transition arises from a rule).  
\end{example}

%
%

\subsection{Respectfulness of contextual closure}

We prove a general respectfulness result of contextual closure w.r.t.\ simulation. 
First we generalise contextual closure as follows~\cite{BonchiPPR17}. 
Given algebras $a \colon H_\Sigma(X) \rightarrow X$ and $b \colon H_\Sigma(Y) \rightarrow Y$,
the \emph{contextual closure} $\ctx_{a,b} \colon \Rel_{X,Y} \rightarrow \Rel_{X,Y}$ is defined by 
$
  \ctx_{a,b}(R) = a^{*} \times b^{*} (\Rel(T_\Sigma)(R)) = \{(a^{*}(u), b^{*}(v)) \mid (u,v) \in \Rel(T_\Sigma)(R)\}
$.
For $X = Y = T_\Sigma(\emptyset)$ and 
$a = b = \kappa_\emptyset \colon H_\Sigma T_\Sigma(\emptyset) \rightarrow T_\Sigma(\emptyset)$, 
$\ctx_{a,b}$ coincides with the contextual closure $\ctx$ of Definition~\ref{def:ctx}.
This allows us to formulate the main result of this section, giving 
sufficient conditions for respectfulness of the contextual closure 
with respect to $s$ from Definition~\ref{def:coalg-sim}.
In fact, this result is slightly more general than needed: 
we will always instantiate $(X,a,f)$ below with a $\lambda$-model.

\begin{theorem}\label{thm:simulation-respectful}
	Suppose that $(X,a,f)$ is an oplax model of a monotone abstract GSOS specification $\lambda$,
	and $(Y,b,g)$ is a lax model. Then $\ctx_{a,b}$ is 
	$s$-respectful. 
 \end{theorem}

\subsection{Application to weak similarity}
\label{sec:coalg-weak}

Let $(X, f)$ be an LTS. Define a new LTS $(X,\overline{f})$ by
$x \xrightarrow{\alpha}_f x'
$ iff $
x \Longrightarrow_f \, \xrightarrow{(\alpha)}_{\overline{f}} \, \Longrightarrow_f x'
$.
We call $(X,\overline{f})$ the \emph{wb-saturation} of $(X,f)$. 
Let $s_{wb} \colon \Rel_X \rightarrow \Rel_X$ be the functional for simulation (Definition~\ref{def:coalg-sim})
between $(X,f)$ and $(X,\overline{f})$. Then $R \subseteq s_{wb}(R)$ precisely if $R$ is a weak
simulation on $(X,f)$. 

\begin{proposition}
	Let $(X,a,f)$ be a model of a positive GSOS specification, 
	and suppose $(X, a, \overline{f})$ is a lax model. Then 
	$\ctx_{a,a}$ is $s_{wb}$-respectful. 
\end{proposition}
The condition of being a lax model is exactly as in Lemma~\ref{lem:superrobweak}.
Hence, the contextual closure of any simply WB-cool GSOS language 
is $s_{wb}$-respectful. To obtain an analogous result for delay similarity, we simply adapt
the saturation to \emph{db-saturation}, and the appropriate functional 
$s_{db}$.

\paragraph{Branching similarity.}

To capture branching simulations of LTSs in the coalgebraic framework,
we will work again with saturation. It is not immediately clear how to do so:
we encode branching simulations by slightly changing the functor, in order
to make relevant intermediate states observable. 

Let $B'(X) = (\pow (X \times X))^{\act}$. A $B'$-coalgebra is similar to an LTS, but transitions take the form
$x \xrightarrow{\alpha} (x', x'')$, i.e., to a pair of next states. We will use this to encode 
branching similarity, as follows. Given an LTS $(X, f)$, define
the \emph{bb-saturation} as the coalgebra $(X, \overline{f})$ where 
$x \xrightarrow{\alpha}_{\overline{f}} (x', x'')$
iff
$x \Rightarrow x' \xrightarrow{(\alpha)} x'' $.
Further, note that every LTS $(X, f)$ gives a $B'$ coalgebra 
$(X, f')$ by setting $x \xrightarrow{a}_{f'} (x', x'')$ iff $x'=x$
  and $ x \xrightarrow{a}_f x''$. 

For an LTS $(X,f)$, consider the functional $s_{bb} \colon \Rel_X \rightarrow \Rel_X$
for $B'$-simulation between $(X,f')$ and $(X, \overline{f})$ (Definition~\ref{def:coalg-sim}). 
Then a relation $R \subseteq X \times X$ is a branching simulation precisely if
$R \subseteq s_{bb}(R)$. 

To obtain the desired respectfulness result from Theorem~\ref{thm:simulation-respectful}, 
the last step is to obtain a GSOS specification for $B'$ from a given positive GSOS
specification (for $B$). 
This is possible if all operators are straight.
In that case, every rule is of the form 
$\frac{\{x_i \xrightarrow{\beta_i} x_i'\}_{i \in I}}
{\sigma(x_1,\dots,x_{\arity(\sigma)}) \step{\alpha} t}
$
for some $I \subseteq \{1, \ldots, \arity(\sigma)\}$. 
This is translated to
$$\frac{\{x_i \xrightarrow{\beta_i} (x_i'',x_i')\}_{i \in I}}
{\sigma(x_1,\dots,x_{\arity(\sigma)}) \step{\alpha} (t^\rho,t)}
\qquad 
\text{ 
where }
\rho(x) = 
\begin{cases}
	x_i'' & \text{ if $x = x_i$ for some $i \in I$} \\
	x & \text{ otherwise}
\end{cases}
$$
If the original specification is presented as an abstract GSOS specification $\lambda$, 
then we denote the corresponding abstract GSOS specification (for $B'$) according
to the above translation by $\lambda'$. 
(It is currently less clear how to represent this translation directly at the abstract level; we leave this
for future work.)


%

\begin{proposition}\label{prop:coalg-bb}
	Let $(X,a,f)$ be a model of a positive GSOS specification $\lambda$ 
	with only straight rules. 
	Then $(X,a,f')$ is a model of $\lambda'$, defined as above; and
	if $(X, a, \overline{f})$ is a lax model, with $(X,\overline{f})$ the bb-saturation, then 
	$\ctx_{a,a}$ is $s_{bb}$-respectful. 
\end{proposition}
We recover Theorem~\ref{thm:br-up-to-ctx}
from Proposition~\ref{prop:coalg-bb} and Lemma~\ref{lem:superrob} (and Lemma~\ref{lem:simulationisenough} to move
 from similarity to bisimilarity). 
Again, to obtain respectfulness for $\eta$-similarity, one simply adapts
the notion of saturation.

%% file: conc-fw.tex
\section{Conclusion and Future Work}\label{sec:conc-fw}

We have seen two main up-to techniques, that can be combined: expansion and, most notably, 
contextual closure. In particular, 
we have shown that for any language defined by a simply cool format, the contextual closure is respectful
for the associated equivalence; this applies to weak, branching, $\eta$ and delay bisimilarity. The latter
follows from a general coalgebraic argument on simulation. 

There are several avenues left for future work. First, we have treated up-to-expansion on a case-by-case basis;
it would be useful to have a uniform treatment of this technique that instantiates to various weak equivalences.
Second, it would be interesting to investigate up-to context for \emph{rooted} and \emph{divergence-sensitive} versions of the weak behavioural equivalences. 
Associated `cool' rule formats have already been proposed~\cite{BloomCT88:GSOS,glabbeek:cool}. 
Third, the current treatment of up-to context heavily relies on positive formats; whether our results
can be extended to rule formats with negative premises is left open. Perhaps the \emph{modal decomposition}
approach to congruence results~\cite{FvG17:dcii,FGL19} can help---investigating the 
relation of this approach to up-to techniques is an exciting direction of research.  Finally, extension of the formats to languages including a recursion construct would be very interesting, especially since the proofs that weak and branching bisimilarity are compatible with this construct use up-to techniques \cite{Mil89,Gla93}.

\paragraph{Acknowledgements.} We thank Filippo Bonchi for the idea
how to encode branching bisimilarity coalgebraically, and the reviewers for their useful comments.

%% file: appendix.tex
\section{Proofs for Section~\ref{sec:prelims}}

\begin{lemma}\label{lem:straightpremises}
Let $\lang=(\sig,\rules)$ be a positive GSOS language and suppose that
all rules in $\rules$ are straight.
Then for all rules $\frac{H}{f(x_1,\dots,x_n) \step{\alpha} t}$ of $\lang$ the
set of premises $H$ is of the form $\{x_i\step{\beta_i}y_i \mid i\in I\}$
for some index set $I \subseteq \{1,\dots,n\}$ with the variables
  in $x_i$ ($i\in I$) and $y_i$ ($i\in I$) all distinct.
\end{lemma}
\begin{proof}
It follows from straightness that each variable $x_i$ can appear at most once
as a left-hand side of a premise.
Since the right-hand sides of the premises should also be distinct,
we can uniquely relate them to their left-hand side.
\end{proof}

  \begin{lemma}\label{lem:auxrobweak}
  Let $\lang$ be a positive GSOS language that satisfies Clause~\ref{item:coolreceivingpatience} of Definition~\ref{def:formats}, let
  $(\states,\act,\trans)$ be a model for $\lang$, and let
  $\eta,\theta:\vars\rightarrow\states$ be assignments.
  If $t$ is the target of a rule, $\theta(y)\tausteps\eta(y)$ for
  every receiving variable $y$ in $t$, and $\theta(x)=\eta(x)$ for
  every variable $x$ in $t$ that is not receiving, then
  $t^{\theta}\tausteps t^{\eta}$.
  \end{lemma}
\begin{proof}
  By a straightforward induction on the structure of $t$ (see \cite{glabbeek:cool}).
\end{proof}

\begin{proof}[Proof of Lemma~\ref{lem:superrobweak}]
  By Lemma~~\ref{lem:straightpremises}, there exists
  $I\subseteq\{1,\dots,n\}$ and variables $x_i$, $y_i$
  ($i\in I$), al distinct, such that
    $H=\{x_i\step{\beta_i}y_i\mid i\in I\}$.
    It follows that if for each premise $x_i\step{\beta}y_i$ in $H$ it holds that
    $\eta(x_i)\tausteps\step{(\beta)}\tausteps\eta(y_i)$, then there
    exists an assignment
      $\theta: \vars\rightarrow\states$
    such that for all $i\in I$ we have
      $\eta(x_i)\tausteps\theta(x_i)\step{(\beta_i)}\theta(y_i)$,
    and for every variable $x\not\in \{x_i, y_i\mid i\in I\}$ we have
    $\eta(x)=\theta(x)$.

    If $\frac{H}{\sigma(x_1,\dots,x_n)\step{\alpha} t}$ is a patience
    rule, then $I=\{i\}$ for some $1\leq i \leq n$, $\beta_i=\tau$ and
    $t=\sigma(x_1,\dots,y_i,\dots,x_n)$.
    With repeated application of the rule we then get that
    $\sigma(x_1,\dots,x_n)^{\eta}
    \tausteps \sigma(x_1,\dots,x_n)^{\theta}
    \step{(\tau)}\sigma(x_1,\dots,y_i,\dots,x_n)^{\theta}
    \tausteps\sigma(x_1,\dots,y_i,\dots,x_n)^{\eta}=t^{\eta}$.

    Otherwise, $\beta_i\neq\tau$ for all $1\leq i \leq n$ by
    Clause~\ref{item:cooltaupatience} of
    Definition~\ref{def:formats}. Since, by
    Clause~\ref{item:coolactive} of Definition~\ref{def:formats}, for
    every active argument there is a patience rule, we have that
    $\sigma(x_1,\dots,x_n)^{\eta}\tausteps\sigma(x_1,\dots,x_n)^{\theta}$,
    and $\sigma(x_1,\dots,x_n)^{\theta}\step{\alpha} t^{\eta}$ by an
    application of the rule. Clause~\ref{item:coolsmooth} then yields
    that $\theta(x)=\eta(x)$ for all variables $x$ in $t$ that are not
    receiving in $t$, so $t^{\theta}\tausteps t^{\eta}$ by Lemma~\ref{lem:auxrobweak}.
\end{proof}

\begin{proof}[Proof of Lemma~\ref{lem:superrob}]
    If $\frac{H}{\sigma(x_1,\dots,x_n)\step{\alpha} t}$ is a patience
    rule, then $I=\{i\}$ for some $1\leq i \leq n$, $\beta_i=\tau$ and
    $t=\sigma(x_1,\dots,y_i,\dots,x_n)$.
    With repeated application of the rule we then get that
    $\sigma(x_1,\dots,x_n)^{\eta}
    \tausteps \sigma(x_1,\dots,x_n)^{\theta}
    \step{(\tau)}\sigma(x_1,\dots,y_i,\dots,x_n)^{\theta}=t^{\theta}$.

    Otherwise, $\beta_i\neq\tau$ for all $1\leq i \leq n$ by
    Clause~\ref{item:cooltaupatience} of
    Definition~\ref{def:formats}. Since, by
    Clause~\ref{item:coolactive} of Definition~\ref{def:formats}, for
    every active argument there is a patience rule, we have that
    $\sigma(x_1,\dots,x_n)^{\eta}\tausteps\sigma(x_1,\dots,x_n)^{\theta}$,
    and $\sigma(x_1,\dots,x_n)^{\theta}\step{\alpha} t^{\theta}$ by an
    application of the rule.
\end{proof}
  
\section{Proofs for Section~\ref{sec:abstractframework}}\label{sec:proof-coalg}

\begin{proof}[Proof of Lemma~\ref{lem:compositionality}]
All listed properties come from \cite{pous:coin,pous:enhancements}.
\begin{enumerate}
\item[1,2.] The companion of $b\meet\id$ is equal to $\bro$,
that is, $\bro = \bigvee\{f\mid f\circ (b\meet\id) \below (b\meet\id) \circ f\}$.
The first two items follow from this fact immediately.

\item[3.] We use that $\bro$ is idempotent.
The third item then follows easily: $g \circ f \below \bro \circ \bro \below \bro$.

\item[4.] Let $\cX$ be a relation. Then:
\[(\utS \circ (\brmono \meet \id))(\cX) = \cS \subseteq \brmono(\cS) = ((\brmono \meet \id) \circ \utS)(\cX)\,.\]
So $\utS$ is $\brmono$-respectful, hence $\utS \below \bro$ by the first item.

\item[5.] We use that $\bro(\emptyset) = \gfp(b)$. Then:
\[f(\gfp(b)) \subseteq f(\bro(\emptyset)) \subseteq \bro(\bro(\emptyset)) \subseteq \bro(\emptyset) = \gfp(b)\,.\]

\item[6.] We have to prove $\gfp(b\circ f)\subseteq \gfp(b)$.
Since $\bro$ is sound we have $\gfp(b \circ f) \below \gfp(b \circ \bro) \below \gfp(b)$.
\end{enumerate}
\end{proof}

\begin{proof}[Proof of Lemma~\ref{lem:reformulation}]
Adapted from a similar proof in \cite{pous:enhancements}.

\noindent $\Rightarrow)$
Suppose that $f$ is $b$-respectful and
consider two arbitrary relations $\cR,\cS$ such that $\cR \subseteq \cS$ and $\cR \subseteq b(\cS)$.
From set theory we can derive that $\cR \subseteq b(\cS) \cap \cS$.
Monotonicity of $f$ gives that $f(\cR) \subseteq f(b(\cS) \cap \cS)$.
By respectfulness we get $f(b(\cS) \cap \cS) \subseteq b(f(\cS)) \cap f(\cS)$.
Set theory then yields the conclusion $f(\cR) \subseteq b(f(\cS))$.

\noindent $\Leftarrow)$
Suppose that $\forall \cR,\cS: \cR \subseteq \cS \text{ and } \cR \subseteq b(\cS) \text{ implies }
f(\cR) \subseteq b(f(\cS))$.

\noindent Let $\cX$ be some relation.
First we prove that $f(b(\cX) \cap \cX) \subseteq b(f(\cX))$.
Define $\cR := b(\cX) \cap \cX$ and $\cS := \cX$.
To fit these relations in the assumption one can call upon set theory to confirm
$\cR \subseteq \cS$ and $\cR \subseteq b(\cS)$.
Then the assumption yields $f(b(\cX) \cap \cX) \subseteq b(f(\cX))$.

Proving that $f(b(\cX) \cap \cX) \subseteq f(\cX)$
is nothing more than the observation of $b(\cX) \cap \cX \subseteq \cX$
and an application of monotonicity of $f$ to this observation.
Hence $f(b(\cX) \cap \cX) \subseteq b(f(\cX)) \cap f(\cX)$,
so $f$ is $b$-respectful.
\end{proof}

\section{Proofs for Section~\ref{sec:bb-up-to}}\label{app:bb-up-to}

\begin{proof}[Proof of \ref{lem:expansionrespectful}]
    Calling upon the formulation of respectfulness from Lemma~\ref{lem:reformulation},
	consider two relations $\cR, \cS$ and suppose that $\cR \subseteq \cS$ and $\cR \subseteq \brmono(\cS)$.
	We prove that ${\eRe} \subseteq \brmono(\eSe)$.
	To this end consider two processes $P,Q$ such that $P \reRe Q$.
    We should prove that $(P,Q)\in\brmono(\reSe)Q$.
	By relation composition there exist $P_0, Q_0$ with $P \rgexp P_0 \rcR Q_0 \rlexp Q$.
    Now suppose that $P \step{\alpha} P'$.
	Since $\gexp$ is a branching expansion one of two cases can occur.
	\begin{itemize}
	\item If $\alpha=\tau$ and $P' \rgexp P_0$, then
		we have $P' \rgexp P_0 \rcR Q_0 \rlexp Q$.
		Since we have $\cR \subseteq \cS$, we obtain $P' \rgexp P_0 \rcS Q_0 \rlexp Q$.
		Then the conclusion for this case follows from the observation that
		$Q \tausteps Q \step{(\alpha)} Q$ and $P' \reSe Q$.
	\item In the other case there exists $P_0'$ with $P' \rgexp P_0'$ and $P_0 \step{\alpha} P_0'$.
		Since $\cR \subseteq \brmono(\cS)$ and $P_0 \rcR Q_0$ there exist
		$Q_0',Q_0''$ with $Q_0 \tausteps Q_0' \step{(\alpha)} Q_0''$ and
		$P_0 \rcS Q_0'$ and $P_0' \rcS Q_0''$.
		Then $Q_0 \tausteps Q_0'$ and $Q_0 \rlexp Q$
		yield some $\hat{Q}$ with $Q \tausteps \hat{Q}$ with $Q_0 \rlexp \hat{Q}$.
		The transition $Q_0' \step{(\alpha)} Q_0''$ yields another case distinction.
		\begin{itemize}
			\item If $\alpha=\tau$ and $Q_0' = Q_0''$ then we pick $Q'=Q''=\hat{Q}$.
				Since $\lexp$ is reflexive it follows that
				$Q \tausteps Q' \step{(\alpha)} Q''$ with
				$P \rgexp P_0 \rcS Q_0' \rlexp Q'$ and $P' \rgexp P_0' \rcS Q_0' \rlexp Q''$.

			\item In the other case we call upon the related pair $Q_0' \rlexp \hat{Q}$
				to give us $Q'$ and $Q''$ such that $\hat{Q} \tausteps Q' \step{\alpha} Q''$
                with $Q_0' \rlexp Q'$ and $Q_0'' \rlexp Q''$.
				Then $Q \tausteps Q'$ follows from transitivity of $\tausteps$.
				Hence in this case we can also conclude that
				$Q \tausteps Q' \step{(\alpha)} Q''$ with
				$P \rgexp P_0 \rcS Q_0' \rlexp Q'$ and $P' \rgexp P_0' \rcS Q_0' \rlexp Q''$.
		\end{itemize}
	\end{itemize}
	In each of the cases we showed the existence of $Q',Q''$ with $Q \tausteps Q' \step{(\alpha)} Q''$
	and $P \reSe Q'$ and $P' \reSe Q''$.
    The case where $Q \step{\alpha} Q'$ is analogous,
	hence ${\eRe} \subseteq \brmono(\eSe)$.
    \qed
\end{proof}

Below, we will make use of the following basic result. 

\begin{proposition}\label{prop:substitutionpreservescontext}
If $\rho(x) \rctxR \theta(x)$ for all variables $x\in\varsof(t)$ then
${t^\rho \rctxR t^\theta}$ for all terms $t$.
\end{proposition}
\begin{proof}
By induction on $t$.
\end{proof}

\begin{proof}[Proof of Theorem~\ref{thm:br-up-to-ctx}]
As in Lemma~\ref{lem:expansionrespectful}, we have to show that
$\cR \subseteq \cS$ and $\cR \subseteq \brmono(\cS)$ implies
$\ctxLR \subseteq \brmono(\ctxL(\cS))$. We proceed by induction on $\ctxLR$.

\noindent\textbf{Base case:}
Notice that $\cS \subseteq \ctxL(\cS)$.
By monotonicity of $\brmono$ we obtain $\brmono(\cS) \subseteq \brmono(\ctxL(\cS))$.
Then the required base case $\cR \subseteq \brmono(\ctxL(\cS))$
follows from the assumption $\cR \subseteq \brmono(\cS)$ and
transitivitiy of $\subseteq$.

\noindent\textbf{Induction step:}
Consider the closed terms $P=f(P_1,\dots,P_n)$ and $Q=f(Q_1,\dots,Q_n)$
with $(P_i,Q_i)\in\ctxLR$ for all $1\leq i\leq n$.
Assume the induction hypothesis $(P_i,Q_i) \in \ctxLR$ implies $(P_i,Q_i)\in\brmono(\ctxL(\cS))$ for all $1\leq i\leq n$.
Suppose that $P\step{\alpha} P'$.
By the semantics there must be a rule $\frac{\{x_i \step{\beta_i}y_i\mid i\in I\}}{f(x_1,\dots,x_n)\step{\alpha}t}$
and a substitution $\rho$ such that
\begin{itemize}
\item $\rho(x_i)=P_i$ for all $1\leq i\leq n$;
\item $t^\rho = P'$; and
\item for all premises $x_i \step{\beta_i} y_i \in H$
it holds that $\rho(x_i) \step{\beta_i} \rho(y_i)$.
\end{itemize}
We have to find $Q', Q''$ such that $Q\tausteps Q'\step{(\alpha)}Q''$ and
$(P,Q'),(P',Q'')\in\ctxL(\cS)$.
To this end we use the induction hypothesis and the index set of the rule
in order to construct substitutions $\eta$ and $\theta$ that satisfy the premise of Lemma~\ref{lem:superrob}.
\begin{itemize}
\item For $i\in I$ there is a premise $x_i\step{\beta_i}y_i \in H$.
By $(P_i,Q_i)\in\ctxLR$ and the induction hypothesis we obtain
$(P_i,Q_i)\in\brmono(\ctxL(\cS))$.
Then since $P_i \step{\beta_i} \rho(y_i)$
there exist $Q_i', Q_i''$ such that
$Q_i \tausteps Q_i' \step{(\beta_i)} Q_i''$ with
$(P_i,Q_i')\in\ctxL(\cS)$ and $(\rho(y_i), Q_i'')\in\ctxL(\cS)$.
We define $\eta(x_i) = Q_i$, $\theta(x_i)=Q_i'$ and $\theta(y_i)=Q_i''$.
\item For $i\notin I$ there is no premise, so we define $\eta(x_i)=\theta(x_i)=Q_i$.
\end{itemize}
These substitutions suffice to apply Lemma~\ref{lem:superrob}.
Let $Q'=f(x_1,\dots,x_n)^\theta$ and $Q''=t^\theta$.
By construction we have $f(x_1,\dots,x_n)^\eta = Q$,
so then the lemma gives us $Q \tausteps Q' \step{(\alpha)} Q''$.

It remains to show that $(P,Q'), (P',Q'')\in\ctxL(\cS)$.
For the first membership we need $(P_i,\theta(x_i))\in\ctxL(\cS)$
for all $1\leq i\leq n$.
This is indeed the case:
\begin{itemize}
\item for $i\in I$ we already deduced
$(P_i,Q_i')\in\ctxL(\cS)$;
\item for $i\notin I$ we use the assumption $(P_i,Q_i)\in\ctxL(\cR)$.
From $Q_i=\eta(x_i)=\theta(x_i)=Q_i'$ we deduce $(P_i,Q_i')\in \ctxL(\cR)$.
Then from $\cR \subseteq \cS$ we can obtain $\ctxL(\cR) \subseteq \ctxL(\cS)$ by monotonicity of $\ctxL$.
Hence $(P_i,Q_i')\in\ctxL(\cS)$.
\end{itemize}
 $(P,Q')\in\ctxL(\cS)$ by definition of contextual closure.

The second membership $(P',Q'')\in\ctxL(\cS)$ requires $(t^\rho,t^\theta)\in\ctxL(\cS)$.
Note that due to the format $\varsof(t) \subseteq \{x_1,\dots,x_n\} \cup \{y_i \mid i \in I\}$.
For substitutions $\rho, \theta$ we already established that $(\rho(x), \theta(x)) \in\ctxL(\cS)$ for all $x\in\varsof(t)$.
Then Proposition~\ref{prop:substitutionpreservescontext} yields $(t^\rho,t^\theta)=(P',Q'')\in\ctxL(\cS)$.
\end{proof}

\section{Examples}\label{app:examples}
The two examples that we consider are about CCS extended with the replication operator
from the $\pi$-calculus.
Some operational rules are the following:
\[
\begin{array}{cccc}
  \texttt{Repl}\ \frac{%
    !P | P \xrightarrow{\alpha} P'
  }{%
    !P \xrightarrow{\alpha} P'
  }
& \qquad
  \texttt{ParL}\ \frac
	{P \xrightarrow{\alpha} P'}
	{P\pc Q\xrightarrow{\alpha} P'\pc Q}
& \qquad
  \texttt{ParL}\ \frac
	{P \xrightarrow{\alpha} P'}
	{P+Q\xrightarrow{\alpha} P'}
& \qquad
  \texttt{Seq}\ \frac
	{}
	{\alpha.P \xrightarrow{\alpha} P}
\end{array}
\]

\begin{proof}[Proof for Example~\ref{ex:uptoeCRe}]
We have to show that $\cR=\{({!\tau.(a\pc\compl{a})},{!(\tau.a+\tau.\compl{a})})\}$
is a branching bisimulation up to $\uteCRe$.
For the left-to-right part of the game suppose that
\[
!\tau.(a\pc\compl{a}) \step{\tau} {!\tau.(a\pc\compl{a}) \pc (a\pc\compl{a})\pc (\tau.(a\pc\compl{a}))^n}
\]
and fix the answer of the right process to be
\[!(\tau.a+\tau.\compl{a}) \tausteps {!(\tau.a+\tau.\compl{a})} \step{(\tau)} {!(\tau.a+\tau.\compl{a})}
\,.
\]
We have to relate two pairs of processes, but the first one is trivial.
By reflexivity of $\gexp$ and extensiveness of $\ctxL$ we immediately get:
$!\tau.(a\pc\compl{a}) \mathrel{\gexp\rctxLR\lexp} {!(\tau.a+\tau.\compl{a})}$.

Relating ${!\tau.(a\pc\compl{a}) \pc (a\pc\compl{a})\pc (\tau.(a\pc\compl{a}))^n}$
to the term $!(\tau.a+\tau.\compl{a})$ by the relation $\gexp\rctxLR\lexp$ requires more work.
To this end we use that $!P\pc P \rsbis !P$ (and hence $!P\pc P \rgexp !P$)
to remove the $n$ parallel copies of $\tau.(a\pc \compl{a})$.
So we have ${!\tau.(a\pc\compl{a}) \pc (a\pc\compl{a})\pc (\tau.(a\pc\compl{a}))^n} \rgexp
{!\tau.(a\pc\compl{a})\pc(a\pc\compl{a})}$.

Now we use the context $\square \pc (a \pc \compl{a})$
and the relation $\cR$ to deduce
${!\tau.(a\pc\compl{a})\pc(a\pc\compl{a})}
\rctxR
{!(\tau.a+\tau.\compl{a})\pc(a\pc\compl{a})}$.
It remains to show that 
${!(\tau.a+\tau.\compl{a})\pc(a\pc\compl{a})} \rlexp {!(\tau.a+\tau.\compl{a})}$.
First we use the expansion property $!P\pc Q \rlexp {!P\pc \tau.Q}$
to obtain ${!(\tau.a+\tau.\compl{a})\pc(a\pc\compl{a})} \rlexp {!(\tau.a+\tau.\compl{a})\pc(\tau.a\pc\tau.\compl{a})}$.
Then, we use the properties $!(P + Q) \rsbis {!(P+Q)\pc P}$ and $P\pc Q \rsbis Q\pc P$ to obtain
${!(\tau.a+\tau.\compl{a})\pc(\tau.a\pc\tau.\compl{a})} \rlexp {!(\tau.a+\tau.\compl{a})}$.
Finally we obtain the desired conclusion by transitivity of $\lexp$.

The $\tau$-transitions of the $\compl{a}$-component of the sum are handled similarly,
so we finished the left-to-right part of the game.
An overview of the proof is given in Figure~\ref{fig:uptoeCRe}.
The right-to-left part of the game is an exercise for the reader.
It then follows that ${!\tau.(a\pc\compl{a})} \rbrbis {!(\tau.a+\tau.\compl{a})}$

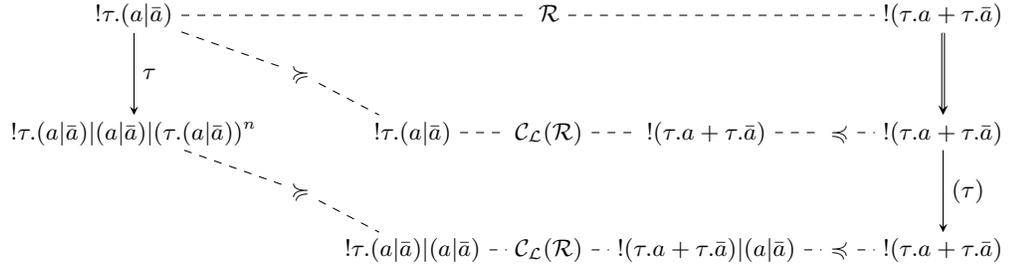
\begin{figure}[!h]
\centering
\input{figureuptoeCRe}
\caption{An overview of Example~\ref{ex:uptoeCRe}.}
\label{fig:uptoeCRe}
\end{figure}
\end{proof}

\begin{proof}[Proof for Example~\ref{ex:uptosCRs}]
We have to show that $\cR=\{(!(a + b), {!\tau.a}\pc{!\tau.b})\}$
is a branching bisimulation up to $\utsCRs$.
For the left-to-right part of the game
suppose that ${!(a+b)} \step{a} {!(a+b)\pc\dl\pc(a+b)^n}$.
and fix the answer of the right process to be
${!\tau.a}\pc{!\tau.b} \step{\tau} {(!\tau.a}\pc a) \pc {!\tau.b} \step{a} {(!\tau.a\pc\dl)\pc !\tau.b}$.
Then we have to relate two pairs by the relation $\sbis\rctxR\sbis$.
\begin{itemize}
\item For the first pair observe that $!(a+b) \rsbis {!(a+b) \pc a}$ and hence
    $!(a+b) \rsbis {!(a+b) \pc a}$.
    Since $!(a+b) \rcR {!\tau.a\pc{!\tau.b}}$ and $\square\pc a$ is a context we can derive
${!(a+b)} \pc a \rctxR {(!\tau.a\pc{!\tau.b}) \pc a}$.
    Lastly we have ${(!\tau.a\pc{!\tau.b}) \pc a} \rsbis (!\tau.a\pc a) \pc {!\tau.b}$ because
    of commutativity and associativity of $\pc$ with respect to $\sbis$.
\item For the second pair we can use that $!P\pc P \rsbis {!P}$ and $P\pc \dl \rsbis \dl$.
    Therefore we just have ${!(a+b)\pc\dl\pc(a+b)^n} \rsbis {!(a+b)}$.
    Since $\ctxL$ is extensive we know $!(a+b) \rctxR {!\tau.a\pc{!\tau.b}}$.
    Then finally ${!\tau.a\pc{!\tau.b}} \rsbis {(!\tau.a\pc\dl)\pc !\tau.b}$ is straightforward.
\end{itemize}
The case for the $b$-transition follows similarly.
In Figure~\ref{fig:uptosCRs} an overview is given;
the right-to-left part of the game is an exercise for the reader.
It then follows that $!(a+b) \rbrbis {!\tau.a\pc{!\tau.b}}$.

\begin{figure}[!h]
\centering
\input{figureutsCRs}
\caption{An overview diagram for Example~\ref{ex:uptosCRs}}
\label{fig:uptosCRs}
\end{figure}
\end{proof}

\newpage
\section{Expansion preorders for $\eta$ and delay bisimilarity}\label{app:expansion}

\subsection{Expansion preorders}

\begin{tabbing}
$\etamonogt(\cR) = \{ (P, Q) \mid$ \= for all $P'$ and all $\alpha$, if $P \step{\alpha} P'$ then\\
    \> \hspace{0.5cm} \= there exists $Q'$ such that $Q \step{(\alpha)} Q'$ and $P' \rcR Q'$; and\\
    \> for all $Q'$ and all $\alpha$, if $Q \step{\alpha} Q'$ then\\
    \> \> there exist $P',P'',P'''$ such that $P \tausteps P' \step{\alpha} P'' \tausteps P'''$\\
    \> \> with $P' \rcR Q$ and $P''' \rcR Q' \}\,.$
\end{tabbing}

\begin{tabbing}
$\dmonogt(\cR) = \{ (P, Q) \mid$ \= for all $P'$ and all $\alpha$, if $P \step{\alpha} P'$ then\\
    \> \hspace{0.5cm} \= there exists $Q'$ such that $Q \step{(\alpha)} Q'$ and $P' \rcR Q'$; and\\
    \> for all $Q'$ and all $\alpha$, if $Q \step{\alpha} Q'$ then\\
    \> \> there exist $P',P''$ such that $P \tausteps P' \step{\alpha} P''$\\
    \> \> with $P'' \rcR Q' \}\,.$
\end{tabbing}

\begin{lemma}
The function $\utetaeRe$ is $\etamono$-respectful and
the function $\utdeRe$ is $\dmono$-respectful.
\end{lemma}
\begin{proof}
Both proofs follow the same routine as the proof of Lemma~\ref{lem:expansionrespectful}.
\end{proof}

\section{Proofs for Section~\ref{sec:coalg}}

In this section we prove Theorem~\ref{thm:simulation-respectful}. To this end, we need additional background and definitions related
to simulations and (lax) relation liftings. 

Throughout this appendix, we fix an ordered functor $B \colon \Set \rightarrow \Set$, and
a functor $H_\Sigma$ defined from a signature $\Sigma$ as in Section~\ref{sec:coalg}.
It is useful to observe that the functor $B \times \Id$ is ordered, with $\sqsubseteq_{BX}$ on
the first component and the diagonal $\Delta_X$ on the second; we denote this order by $\sqsubseteq_{BX} \times \Delta_X$
(note that this is not Cartesian product). 
Thus, $(u,x) \sqsubseteq_{BX} \times \Delta_X (v,y)$ iff $u \sqsubseteq_{BX} v$ and $x=y$.

For $(X,f)$ and $(Y,g)$ coalgebras, a \emph{lax coalgebra homomorphism} is a map
$h \colon X \rightarrow Y$ such that $Bh \circ f \sqsubseteq_{BY} g
\circ h$, which is diagrammatically depicted as
$$
\xymatrix@C=1.3cm{
	X \ar[d]_{f} \ar[r]^h  \ar@{}[dr]|{\sqsubseteq_{BY}}
		& Y \ar[d]^g \\
	BX \ar[r]_{Bh}
		& BY
}
$$
If instead $Bh \circ f \sqsupseteq_{BY} g \circ h$ then $h$ is called an \emph{oplax coalgebra homomorphism}.

\begin{example}
	Let $BX = (\pow X)^{\act}$, ordered as usual (Example~\ref{ex:ordered-lts}).
	Given coalgebras $(X,f)$ and $(Y,g)$, a
	lax coalgebra homomorphism is a functional simulation (from left to right): a map $h \colon X \rightarrow Y$
	such that, if $x \xrightarrow{\alpha} x'$ then $h(x) \xrightarrow{\alpha} h(x')$. 
	Instead $h$ is an oplax homomorphism if $h(x) \xrightarrow{\alpha} y'$ implies $\exists x'. \, x \xrightarrow{\alpha} x'$
	and $h(x') = y'$. 
\end{example}

Simulations can equivalently be characterised as follows, which is useful in proofs. 
\begin{lemma}\label{lm:spans}
	For $B$-coalgebras $(X,f)$ and $(Y,g)$, and relations $R,S \subseteq X \times Y$, 
	we have that $R \subseteq s(S)$ iff there exists a map $r \colon R \rightarrow B(S)$
	such that $f \circ \pi_1^R \sqsubseteq_{BX} B\pi_1^S \circ r$ and
	$B\pi_2^S \circ r \sqsubseteq_{BY} g \circ \pi_2^R$:
	$$
	\xymatrix@C=1.3cm{
		X \ar[d]_{f} \ar@{}[dr]|{\sqsubseteq_{BX}}
			& R \ar[l]_{\pi_1^R} \ar[d]^r \ar[r]^{\pi_2^R} \ar@{}[dr]|{\sqsubseteq_{BY}}
			& Y \ar[d]^g \\
		BX 
			& BS \ar[l]^{B\pi_1^S} \ar[r]_{B\pi_2^S}
			& BY
	}
	$$
	In particular, $R$
	is a simulation iff there exists a coalgebra structure $r \colon R \rightarrow BR$ such that
	the projections $\pi_1, \pi_2$ are, respectively, an oplax and a lax coalgebra homomorphism.
\end{lemma}
\begin{proof}
	Spelling out the definition of $s$, we have that $R \subseteq s(S) = (f \times g)^{-1}(\sqsubseteq_{BX} \relcmp \Rel(B)(S) \relcmp \sqsubseteq_{BY})$ if and only if for each $(x,y) \in R$, there exists $z \in B(S)$ 
	such that $f(x) \sqsubseteq_{BX} B\pi_1^S(z) \, \Rel(B)(S) B\pi_2^S \sqsubseteq_{BY} g(y)$. 
	The equivalence follows easily: if $R \subseteq s(R)$ then, for every pair $(x,y) \in R$, 
	choose such an element $z$ and let $r(x,y) = z$. Conversely, if $r \colon R \rightarrow B(S)$
	satisfies the required properties, choose $z$ to be $r(x,y)$ for each $(x,y) \in R$ to show
	$R \subseteq s(S)$. 
\end{proof}

In the theory of abstract GSOS and bialgebras, it is well-known that 
an abstract GSOS specification $\lambda \colon H_\Sigma (B \times \Id) \Rightarrow B T_\Sigma$
gives rise to a distributive law $\overline{\lambda} \colon T_\Sigma (B \times \Id) \Rightarrow (B \times \Id) T_\Sigma$
of the free monad $T_\Sigma$ over the copointed functor $B \times \Id$, meaning that $\lambda$ satisfies 
certain axioms related to the monad structure and the projection $B \times \Id \rightarrow \Id$. 
In fact, this correspondence is one-to-one.

Further, we have that $(X,a,f)$ is a model of $\rho$ iff $(X, a^{*}, \langle f, \id \rangle)$ is a \emph{$\lambda$-bialgebra},
which means that the following diagram commutes:
\begin{equation}\label{eq:bialg}
\xymatrix@C=1.3cm{
T_\Sigma(X) \ar[d]_{a^{*}} \ar[r]^-{H_\Sigma \langle f, \id \rangle}
	& T_\Sigma (BX \times X) \ar[r]^-{\overline{\lambda}_X}
	& B T_\Sigma(X) \times T_\Sigma(X) \ar[d]^{B a^{*} \times a^{*}} \\
X \ar[rr]^-{\langle f, \id \rangle} 
	& & BX \times X
}
\end{equation}

Following from the results of~\cite{BonchiPPR17}, there is the following `lax' version of the above. 
First, note that both the domain and codomain of a component $\overline{\lambda}_X$ carry a pre-order: 
$\Rel(H_\Sigma)(\sqsubseteq_{BX} \times \Delta_X)$ and $\sqsubseteq_{BT_\Sigma} \times \Delta_{T_\Sigma}$, respectively.
We say $\overline{\lambda}$ is \emph{monotone} if it is monotone with respect to these orders. 
Further, we say $(X,a^{*}, \langle f, \id \rangle)$ is a \emph{lax bialgebra} 
if the diagram~\ref{eq:bialg} commutes up to $\sqsubseteq_{BX} \times \Delta_X$, that is,
$\langle f, \id \rangle \circ a^{*} \mathrel{\sqsubseteq_{BX} \times \Delta_X} {(B \times \Id)a^{*} \circ \overline{\lambda}_X \circ T_\Sigma \langle f, \id \rangle}$. It is an \emph{oplax} bialgebra 
if it is lax for the converse order, i.e., $\langle f, \id \rangle \circ a^{*} \mathrel{\sqsupseteq_{BX} \times \Delta_X} {(B \times \Id)a^{*} \circ \overline{\lambda}_X \circ T_\Sigma \langle f, \id \rangle}$.

\begin{lemma}\label{lm:lax-model-to-bialg}
	If $\lambda \colon H_\Sigma (B \times \Id) \Rightarrow B T_\Sigma$ is monotone, 
	then $\overline{\lambda}$ is monotone as well. 
	Moreover, if $(X,a,f)$ is a lax model then $(X, a^{*}, \langle f, \id)$ is a lax bialgebra. 
\end{lemma}
\begin{proof}
	This follows from Lemma 11.2 and Lemma 11.5 in~\cite{BonchiPPR17}.
\end{proof}
Of course, analogous results hold for oplax models, which follows simply by reversing the order. 

Further, we will use a basic property of relation lifting.
\begin{lemma}\label{lm:locally-monotone}
	For any functor $B \colon \Set \rightarrow \Set$, 
	relation lifting $\Rel(B)$ is locally monotone, that is, if $f,g \colon X \rightarrow BY$ 
	are such that $f \sqsubseteq_{BY} g$ then $Bf \mathrel{\Rel(B)(\sqsubseteq_{BY})} Bg$. 
\end{lemma}

Now, we are ready to prove the main result on respectfulness. 

\begin{proof}[Proof of Theorem~\ref{thm:simulation-respectful}]
 	Suppose $R \subseteq s(S)$ and $R \subseteq S$. By Lemma~\ref{lem:reformulation}, it suffices to prove
 	that $\ctx_{a,b}(R) \subseteq \ctx_{a,b}(S)$. 
 	
 	First, by Lemma~\ref{lm:spans}, we have $r \colon R \rightarrow B(S)$ 
 	such that $f \circ \pi_1^R \sqsubseteq_{BX} B\pi_1^S \circ r$ and
	$B\pi_2^S \circ r \sqsubseteq_{BY} g \circ \pi_2^R$:
	\begin{equation}\label{eq:progress1}
	\xymatrix@C=1.3cm{
		X \ar[d]_{f} \ar@{}[dr]|{\sqsubseteq_{BX}}
			& R \ar[l]_{\pi_1^R} \ar[d]^r \ar[r]^{\pi_2^R} \ar@{}[dr]|{\sqsubseteq_{BY}}
			& Y \ar[d]^g \\
		BX 
			& BS \ar[l]^{B\pi_1^S} \ar[r]_{B\pi_2^S}
			& BY
	}
	\end{equation}
	By assumption $R \subseteq S$; call the inclusion map $i \colon R \hookrightarrow S$. 
	Together with~\eqref{eq:progress1} we get:
	\begin{equation}\label{eq:progress2}
	\xymatrix@C=1.5cm{
		X \ar[d]_{\langle f, \id \rangle} \ar@{}[dr]|{\sqsubseteq_{BX} \times \Delta_X}
			& R \ar[l]_{\pi_1^R} \ar[d]|{\,\langle r, i \rangle} \ar[r]^{\pi_2^R} \ar@{}[dr]|{\sqsubseteq_{BY} \times \Delta_Y}
			& Y \ar[d]^{\langle g, \id \rangle} \\
		BX \times X 
			& BS \times S \ar[l]^{(B \times \Id) \pi_1^S} \ar[r]_{(B \times \Id) \pi_2^S}
			& BY \times Y
	}
	\end{equation}
	We apply $T_\Sigma$ to the above diagram to get the middle part of the following, which
	commutes up to $\Rel(T_\Sigma)(\sqsubseteq_{BX} \times \Delta_X)$ respectively 
	$\Rel(T_\Sigma)(\sqsubseteq_{BY} \times \Delta_Y)$ by Lemma~\ref{lm:locally-monotone}:
	$$
	{\small
	\xymatrix@C=1.3cm{
		X \ar[dd]^f \ar@{}[ddr]|{\sqsubseteq_{BX} \times \Delta_X}
			& T_\Sigma X \ar[d]_{T_\Sigma\langle f, \id \rangle} \ar@{}[dr]|{\Rel(T_\Sigma)(\sqsubseteq_{BX} \times \Delta_X)}
			\ar[l]_{a^{*}}
			& T_\Sigma R \ar[l]_{T_\Sigma\pi_1^R} \ar[d]|{\,\langle r, i \rangle} \ar[r]^{T_\Sigma \pi_2^R} \ar@{}[dr]|{\Rel(T_\Sigma)(\sqsubseteq_{BY} \times \Delta_Y)}
			& T_\Sigma Y \ar[d]^{T_\Sigma\langle g, \id \rangle} 
				\ar[r]^{b^{*}} \ar@{}[ddr]|{\sqsubseteq_{BY} \times \Delta_Y}
			& Y \ar[dd]_g 
			\\
			& T_\Sigma BX \times X \ar[d]^{\overline{\lambda}_X} 
			& T_\Sigma BS \times S \ar[l]^{T_\Sigma (B \times \Id) \pi_1^S} \ar[r]_{T_\Sigma(B \times \Id) \pi_2^S}
				\ar[d]^{\overline{\lambda}_S} 
			& T_\Sigma BY \times Y 
				\ar[d]^{\overline{\lambda}_Y} 
			& 
			\\
		(B \times \Id) X
			& (B \times \Id)T_\Sigma X \ar[l]^{(B \times \Id)b^{*}}
			& (B \times \Id)T_\Sigma S \ar[l]^{(B \times \Id)T_\Sigma \pi_1^S}  \ar[r]_{(B \times \Id)T_\Sigma \pi_2^S}
			& (B \times \Id)T_\Sigma Y \ar[r]_{(B \times \Id)b^{*}}
			& (B \times \Id)Y
	}
	}
	$$
	The lower squares commute by naturality of $\overline{\lambda}$. 
	The left and right lax commutativity follow from the assumption on $(X,a,f)$ and $(Y,b,g)$ being
	oplax and lax respectively, together with Lemma~\ref{lm:lax-model-to-bialg}.
	Since, further, $\overline{\lambda}_X$ is monotone 
	(as well as $(B \times \Id)\pi_1^S$ and $(B \times \Id)\pi_2^S$ by virtue of $B \times \Id$ being ordered) 
	we can compose these inequalities (and equalities) to obtain
	those in the following diagram. 
	$$
	\xymatrix@C=1.3cm{
		X \ar[d]^f \ar@{}[drr]|{\sqsubseteq_{BX} \times \Delta_X}
			& T_\Sigma X 
			\ar[l]_{a^{*}}
			& T_\Sigma R \ar[l]_{T_\Sigma\pi_1^R} \ar[d]|{\overline{\lambda}_X \circ \langle r, i \rangle} \ar[r]^{T_\Sigma \pi_2^R} \ar@{}[drr]|{\sqsubseteq_{BY} \times \Delta_Y}
			& T_\Sigma Y 
				\ar[r]^{b^{*}}
			& Y \ar[d]_g 
			\\
		(B \times \Id) X
			& (B \times \Id)T_\Sigma X \ar[l]^{(B \times \Id)b^{*}}
			& (B \times \Id)T_\Sigma S \ar[l]^{(B \times \Id)T_\Sigma \pi_1^S}  \ar[r]_{(B \times \Id)T_\Sigma \pi_2^S}
			& (B \times \Id)T_\Sigma Y \ar[r]_{(B \times \Id)b^{*}}
			& (B \times \Id)Y
	}
	$$
	To arrive at the contextual closure, we need to take the direct images along $a^{*} \circ T_\Sigma \pi_1$ etc. 
	To this end, note that there is a surjective map $q_R \colon T_\Sigma R \rightarrow \ctx_{a^{*}, b^{*}}(R)$
	such that $\pi_1 \circ q_R = a^{*} \circ T_\Sigma \pi_1^R$ and 
	$\pi_2 \circ q_R = b^{*} \circ T_\Sigma \pi_2^R $, with $\pi_1$, $\pi_2$ the projections of $\ctx_{a^{*},b^{*}}(R)$
	(and a similar map $q_S$ for $S$). 
	The map $q_R$ has a right inverse $q_R^{-1}$, i.e., $q_R \circ q_R^{-1} = \id_{T_\Sigma R}$. 
	Note that $\pi_1  = \pi_1 \circ q_R \circ q_R^{-1} = a^{*} \circ T_\Sigma \pi_1^R \circ q_R^{-1}$
	and $\pi_2  = \pi_2 \circ q_R \circ q_R^{-1} = b^{*} \circ T_\Sigma \pi_2^R \circ q_R^{-1}$.

	Thus, the upper and lower parts in the diagram below commute:
	$$
	\xymatrix@C=1.3cm{
			& & \ctx_{a^{*},b^{*}}(R) \ar@/_10pt/[dll]_{\pi_1} \ar[d]|{q_R^{-1}} \ar@/^10pt/[drr]^{\pi_2} & &  \\
		X \ar[d]^f \ar@{}[drr]|{\sqsubseteq_{BX} \times \Delta_X}
			& T_\Sigma X 
			\ar[l]_{a^{*}}
			& T_\Sigma R \ar[l]_{T_\Sigma\pi_1^R} \ar[d]|{\overline{\lambda}_X \circ \langle r, i \rangle} \ar[r]^{T_\Sigma \pi_2^R} \ar@{}[drr]|{\sqsubseteq_{BY} \times \Delta_Y}
			& T_\Sigma Y 
				\ar[r]^{b^{*}}
			& Y \ar[d]_g 
			\\
		(B \times \Id) X
			& (B \times \Id)T_\Sigma X \ar[l]^{(B \times \Id)b^{*}}
			& (B \times \Id)T_\Sigma S \ar[l]^{(B \times \Id)T_\Sigma \pi_1^S}  \ar[r]_{(B \times \Id)T_\Sigma \pi_2^S}
			 \ar[d]|{(B \times \Id)q_S}
			& (B \times \Id)T_\Sigma Y \ar[r]_{(B \times \Id)b^{*}}
			& (B \times \Id)Y \\
			& & (B \times \Id)\ctx_{a^{*},b^{*}}(S) \ar@/^10pt/[ull]^{(B\times \Id)\pi_1} \ar@/_10pt/[urr]_{(B \times \Id)\pi_2} & &  
	}
	$$	It follows from Lemma~\ref{lm:spans} that $\ctx_{a^{*},b^{*}}(R) \subseteq s'(\ctx_{a^{*},b^{*}}(S))$, where $s'$ is the functional
	associated to the ordered functor $B \times \Id$ (via $\sqsubseteq_{BX} \times \Delta_X$). 
	It is easy to establish (e.g.,~\cite{rot17}) that this is equivalent 
	to $\ctx_{a^{*},b^{*}}(R) \subseteq s(\ctx_{a^{*},b^{*}}(S))$ and $\ctx_{a^{*},b^{*}}(R) \subseteq \ctx_{a^{*},b^{*}}(S)$;
	the former is what is needed. 
 \end{proof}

%% file: figureuptoeCRe.tex
\begin{tikzpicture}[description/.style={fill=white,inner sep=2pt}]
\matrix (m) [matrix of math nodes, row sep=1em,
column sep=0.75em, text height=1.5ex, text depth=0.25ex]
{
 !\tau.(a\pc\compl{a})                                              &       &                                         & \cR    &                                         &       & !(\tau.a+\tau.\compl{a})  \\
                                                                    & \gexp &                                         &        &                                         &       &                           \\
 !\tau.(a\pc\compl{a})\pc(a\pc\compl{a})\pc(\tau.(a\pc\compl{a}))^{n} &       & !\tau.(a\pc\compl{a})                   & \ctxLR & !(\tau.a+\tau.\compl{a})                & \lexp & !(\tau.a+\tau.\compl{a})  \\
                                                                    & \gexp &                                         &        &                                         &       &                           \\
                                                                    &       & !\tau.(a\pc\compl{a})\pc(a\pc\compl{a}) & \ctxLR & !(\tau.a+\tau.\compl{a})\pc(a\pc\compl{a}) & \lexp & !(\tau.a+\tau.\compl{a})  \\
};

\path[-stealth]
	(m-1-7) edge [double]			   (m-3-7)
	(m-1-1) edge node [right] {$\tau$} (m-3-1)
	(m-3-7) edge node [right] {$(\tau)$} (m-5-7)
;
\path
	(m-1-1) edge [dashed] (m-1-4)
	(m-1-4) edge [dashed] (m-1-7)

	(m-3-1) edge [dashed] (m-4-2)
	(m-4-2) edge [dashed] (m-5-3)
	(m-5-3) edge [dashed] (m-5-4)
	(m-5-4) edge [dashed] (m-5-5)
	(m-5-5) edge [dashed] (m-5-6)
	(m-5-6) edge [dashed] (m-5-7)

	(m-1-1) edge [dashed] (m-2-2)
	(m-2-2) edge [dashed] (m-3-3)
	(m-3-3) edge [dashed] (m-3-4)
	(m-3-4) edge [dashed] (m-3-5)
	(m-3-5) edge [dashed] (m-3-6)
	(m-3-6) edge [dashed] (m-3-7)
;
\end{tikzpicture}

%% file: figureutsCRs.tex
\begin{tikzpicture}[description/.style={fill=white,inner sep=2pt}]
\matrix (m) [matrix of math nodes, row sep=1em,
column sep=0.75em, text height=1.5ex, text depth=0.25ex]
{
 !(a+b)                 &        &             & \cR       &                            &        & !\tau.a\pc{!\tau.b}                      \\
                        & \sbis  &             &           &                            &        &                 \\
 !(a+b)\pc\dl\pc(a+b)^n &        & !(a+b)\pc a & \ctx(\cR) & (!\tau.a\pc{!\tau.b})\pc a & \sbis  & (!\tau.a\pc a) \pc {!\tau.b}        \\
                        & \sbis  &             &           &                            &        &                               \\
                        &        & !(a+b)      & \ctx(\cR) & !\tau.a\pc{!\tau.b}        & \sbis  & (!\tau.a\pc\dl)\pc{!\tau.b}                              \\
};

\path[-stealth]
	(m-1-7) edge [double]			   (m-3-7)

	(m-1-1) edge node [right] {$a$} (m-3-1)
	(m-3-7) edge node [right] {$(a)$} (m-5-7)
;
\path
	(m-1-1) edge [dashed] (m-1-4)
	(m-1-4) edge [dashed] (m-1-7)

	(m-1-1) edge [dashed] (m-2-2)
	(m-2-2) edge [dashed] (m-3-3)
	(m-3-3) edge [dashed] (m-3-4)
	(m-3-4) edge [dashed] (m-3-5)
	(m-3-5) edge [dashed] (m-3-6)
	(m-3-6) edge [dashed] (m-3-7)

	(m-3-1) edge [dashed] (m-4-2)
	(m-4-2) edge [dashed] (m-5-3)
	(m-5-3) edge [dashed] (m-5-4)
	(m-5-4) edge [dashed] (m-5-5)
	(m-5-5) edge [dashed] (m-5-6)
	(m-5-6) edge [dashed] (m-5-7)
        
;
\end{tikzpicture}